\documentclass[sort
]{ceurart}

\usepackage{listings}
\usepackage{amssymb}
\usepackage{amsmath}
\usepackage{amsthm}
\usepackage{booktabs}
\usepackage{subcaption}
\usepackage{algorithm}
\usepackage{algpseudocode}

\begin{document}

\copyrightyear{2022}
\copyrightclause{Copyright for this paper by its authors.
  Use permitted under Creative Commons License Attribution 4.0
  International (CC BY 4.0).}

\conference{SEBD 2026}

\title{Project resilience as network robustness}

\tnotemark[1]
\tnotetext[1]{This paper is a discussion paper describing the research originally presented in \cite{piccolo2018design,piccolo2024evaluating,piccolo2026busfactor}.}

\newtheorem{theorem}{Theorem}
\newtheorem{proposition}{Proposition}
\newtheorem{corollary}{Corollary}
\newtheorem{lemma}{Lemma}[theorem]
\newtheorem{definition}{Definition}

\author[1]{Sebastiano A. Piccolo}[%
orcid=0000-0002-6986-3344,
email=sebastiano.piccolo@unical.it,
]
\cormark[1]
\address[1]{University of Calabria, Department of Mathematics and Computer Science (DeMaCS)}

\author[1]{Giorgio Terracina}[%
orcid=0000-0002-3090-7223,
email=giorgio.terracina@unica.it,
]
% \address[2]{DeMaCS}
% \address[3]{Vrije Universiteit Amsterdam, De Boelelaan 1105, 1081 HV Amsterdam, The Netherlands}

\cortext[1]{Corresponding author.}
% \fntext[1]{These authors contributed equally.}

\begin{abstract}
  Engineering projects are the result of the combined effort of their members.
  Yet, it has been documented that labor division withing projects is unevenly distributed: some project members are specialists undertaking only few tasks, whereas other are generalists and are responsible for the success of many tasks. Moreover, the latter are often facilitators of project integration.
  Such a workload distribution prompts one question: \emph{how resilient is a project to key personnel loss?}
  Far from being a theoretical problem, the reliance of a project on a few key people can lead to severe economic losses and delays.
  We argue that current methods to estimate such a risk are unsatisfactory: some methods offer a best-case estimate and are, therefore, too optimistic; other methods fail to capture project fragmentation leading to biased estimates and unrealistic consequences in many settings.
  In this paper, we develop a novel method to assess project vulnerability by looking at it from the lens of network robustness.
  We compare our method against existing alternatives and show that it offers better and more consistent estimates of project resilience to personnel loss.
\end{abstract}


\maketitle

\section{Introduction}
Engineering projects are a quintessential example of collective problem solving, combining the efforts of professionals with diverse skill-sets.
Due to the nature of the project itself, the diversity of technical skills involved, and the heterogeneity of requirements, the labor distribution observed in projects is highly skewed: a relatively small number of professionals undertake a large number of tasks, while the vast majority are responsible for a limited number of tasks~\cite{yamashita2015pareto,agrawal2018hero,piccolo2018design,majumder2019software}. 
These types of professionals have respectively been termed \emph{generalists} and \emph{specialists}~\cite{piccolo2018design, piccolo2019iterations}.

Such a workload distribution may render a project vulnerable to the loss of key personnel, with heavy consequences including delays, economic losses, and service or infrastructure disruptions that propagate beyond the project itself~\cite{zanetti2013rise,piccolo2018design,russo2024shock}. 
This vulnerability has been named \emph{Bus Factor} (or Truck Factor), informally defined as the minimum number of professionals whose departure from a project would stall it~\cite{zazworka2010developers,cosentino2015assessing,avelino2016novel}.

There is ample evidence of the Bus Factor and its severe consequences. 
The \texttt{left-pad} incident is emblematic: a developer unpublished a small package (\texttt{left-pad}) of just 11 lines of code, breaking the dependencies of a substantial part of the JavaScript ecosystem (including Babel, Webpack, and React)\footnote{More details on the \texttt{left-pad} incident at: \url{https://en.wikipedia.org/wiki/Npm_left-pad_incident}}. 
Additionally, studies on open source communities have shown that when a central contributor leaves, community cohesion, performance, and bug handling suffer dramatically until others step up~\cite{zanetti2013rise,russo2024shock}.
Furthermore, it has been found that core developers (termed \emph{heroes}) are often more productive and introduce fewer bugs than others~\cite{agrawal2018hero,majumder2019software}.

In light of the serious damage a project with a low Bus Factor can incur, it is of paramount importance to have reliable methods to evaluate robustness~\cite{zazworka2010developers, avelino2016novel, piccolo2018design, piccolo2024evaluating}. 
A good measure of Bus Factor not only estimates vulnerability but can be used to drive intervention strategies and optimize the allocation of people to tasks~\cite{piccolo2024evaluating}.

Unfortunately, as we show, the available methods to estimate Bus Factor suffer from severe shortcomings: \emph{1)} they are often not generally applicable, being tailored to specific metadata (e.g., GitHub); \emph{2)} they depend on arbitrary thresholds to define stalling conditions; \emph{3)} they fail to capture the resulting project fragmentation when key personnel leave; and \emph{4)} they imply consequences that contrast with empirical findings from project management (for instance, overestimating the impact of experts while downplaying the role of generalists).

\begin{figure*}
    \centering
    \includegraphics[width=\linewidth]{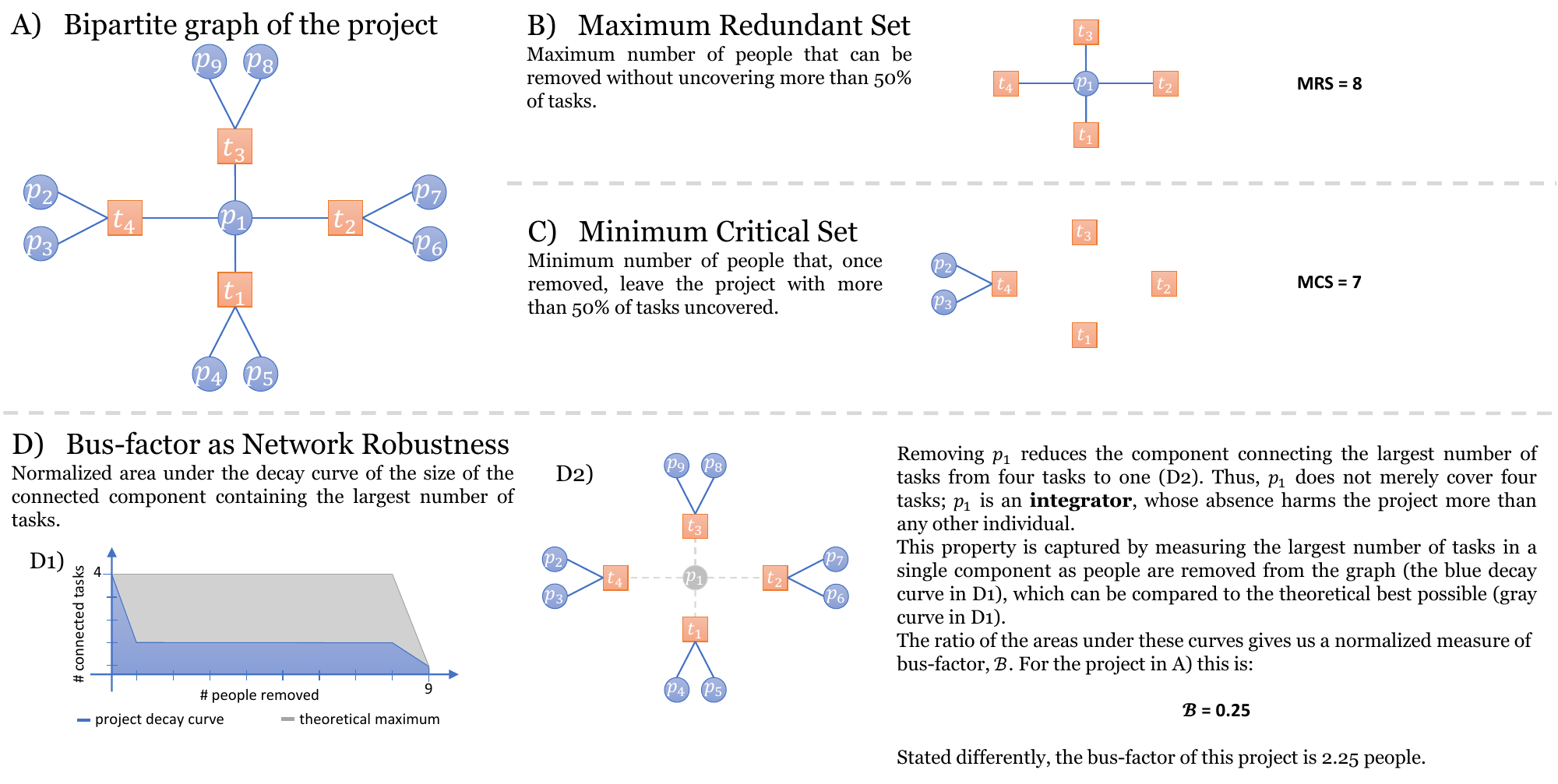}
    \caption{Comparison of the Bus Factor measures formalized as combinatorial problems on bipartite graphs.}
    \label{fig:bus_factor_framework}
\end{figure*}

In this paper, we present a novel measure of Bus Factor, that is domain-independent, threshold-free, and able to account for project fragmentation. 
We model a project as a bipartite graph connecting professionals to their tasks, formulating the Bus Factor as a combinatorial problem.

This modeling approach allows us to generalize existing Bus Factor estimation procedures. 
We formulate them as two complementary combinatorial problems: the Maximum Redundant Set (\textsc{MRS}), seeking the maximum set of people that can be safely removed; and the Minimum Critical Set (\textsc{MCS}), seeking the minimum set of people whose removal uncovers a target number of tasks. 
Consequently, we formulate our new measure as a third combinatorial problem, \textsc{Robustness}, which seeks the removal order of people that minimizes the area under the curve describing the size of the largest connected task component. 
As a result, we make all Bus Factor measures comparable by fixing their input graph. Figure~\ref{fig:bus_factor_framework} provides an intuitive comparison of the Bus Factor measures.

% We model a project as a bipartite graph between professionals and the task they work on and formulate the Bus Factor in a domain-agnostic way, as a combinatorial problem on such a graph.
% This modeling approach allows us to detach existing Bus Factor estimation procedures from their domains, formulating them as two complementary combinatorial problems: the Maximum Redundant Set, \textsc{MRS}, that seeks the maximum set of people that can be safely removed from the graph without uncovering more than a given number of tasks; and the Minimum Critical Set, \textsc{MCS}, that seeks the minimum set of people whose removal uncovers at least a given number of tasks.
% Consequently, we formulate our new Bus Factor measure as third combinatorial problem on this bipartite graph, \textsc{Robustness}, that seeks the removal order of people that minimizes the area under the curve that describes the largest number of tasks connected into a single component.
% As a result, we make all Bus Factor measures comparable, by fixing their input graph (Figure~\ref{fig:bus_factor_framework} offers a visual comparison of the three measures, along with an understanding of the shortcomings of existing approaches).

% Beyond introducing a unified theoretical framework for the computation of the Bus Factor and a new measure grounded in network robustness, we establish the hardness of all Bus Factor measures, propose efficient approximation algorithms, and evaluate the performance of the Bus Factor measures by building a systematic sensitivity analysis guided by project management theory.

\paragraph{Summary of contributions.} We make the following contributions:
\begin{enumerate}
    \item We model projects as bipartite graphs and derive a unifying framework that formulates Bus Factor measures as combinatorial problems of node removal.
    \item We formulate a novel measure of Bus Factor that overcomes the shortcomings of existing formulations.
    \item We prove that all Bus Factor formulations are NP-Hard and develop efficient approximation algorithms for each.
    \item We provide an extensive analysis and comparison of the Bus Factor measures through a series of tests guided by project management theory.
    \item We apply our measure to a real-world project to guide a re-allocation of people to tasks, improving its Bus Factor by 40\%.
\end{enumerate}

Together, our contributions provide both theoretical foundations and scalable algorithms for analyzing and mitigating project vulnerability.

\section{Related Work}
\paragraph{The impact of low Bus Factor.}
Ample evidence suggests that projects with low Bus Factor face severe risks~\cite{zazworka2010developers,cosentino2015assessing,avelino2016novel, ferreira2020turnover,foucault2015impact,avelino2019abandonment}.
Empirical studies consistently show that core developers (termed \emph{heroes}) are more productive and introduce fewer bugs~\cite{agrawal2018hero,majumder2019software}.
Crucially, the turnover of these key personnel is followed by a reduction in code quality and slower bug fixing~\cite{zanetti2013rise,foucault2015impact,avelino2019abandonment,ferreira2020turnover,russo2024shock}.
In extreme cases, projects fail to survive the loss entirely~\cite{avelino2016novel}.

\paragraph{Prior Bus Factor assessment methods.}
Bus Factor estimation has been researched most prominently in software engineering~\cite{zazworka2010developers, yamashita2015pareto, cosentino2015assessing, avelino2016novel, jabrayildaze2022bus}.
Despite its intuitive definition, prior work differs substantially in how projects are modeled and how stalling is defined.
Zazworka \textit{et al.}~\cite{zazworka2010developers} introduced a family of coverage-based measures, defining a project as safe if a given percentage of files (e.g., 50\%) remains covered after contributor removal.
While conceptually appealing, this approach does not explicitly model developer knowledge~\cite{ricca2011difficulty, ferreira2017comparison}.

Subsequent work refined knowledge inference.
Cosentino \textit{et al.}~\cite{cosentino2015assessing} introduced primary and secondary developers based on contribution thresholds.
Most notably, Avelino \textit{et al.}~\cite{avelino2016novel} proposed a greedy heuristic that iteratively removes the most knowledgeable developer (measured via Degree of Authorship) until more than 50\% of files become abandoned.
This heuristic has become the standard in follow-up studies~\cite{jabrayildaze2022bus, ferreira2017comparison}.

\paragraph{Critical Analysis.}
Current approaches mix two distinct concerns: \emph{1)} inferring developer knowledge from metadata (the \emph{filtering} stage); and \emph{2)} identifying the set of core developers (the \emph{estimation} stage).
By conflating these stages, prior methods rely on domain-specific heuristics rather than formal combinatorial optimization.
This realization drives our proposal: a domain-agnostic framework that decouples estimation from filtering, allowing for a rigorous graph-theoretic treatment of the Bus Factor.

\section{A domain-agnostic framework}
We model a project as a bipartite graph $B = (P, T, E)$, where $P$ is the set of professionals (people), $T$ is the set of tasks (or artifacts), and an edge $(p, t) \in E$ exists if professional $p$ contributes to task $t$.
This set $T$ is semantic-agnostic: it can represent source files, functions, or Jira tickets, provided that all elements in $T$ share the same granularity.

\paragraph{Notation and preliminaries.}
For a graph $G=(V, E)$, the \emph{neighborhood} of a vertex $v \in V$, denoted $N(v)$, is the set of vertices adjacent to $v$, and the degree of $v$ is $\deg(v) = |N(v)|$.
In our bipartite context, for a professional $p \in P$, $N(p) \subseteq T$ is the set of tasks they work on. 
Conversely, for a task $t \in T$, $N(t) \subseteq P$ is the set of professionals capable of performing it.
We extend this notation to sets: for a subset of professionals $S \subseteq P$, the \emph{coverage} of $S$, denoted $\operatorname{Cov}(S)$, is the set of all tasks performed by at least one person in $S$:
 $\operatorname{Cov}(S) = \bigcup_{p \in S} N(p) $.

A task $t$ is considered \emph{active} (or covered) with respect to a team $S$ if $t \in \operatorname{Cov(S)}$; otherwise, it is \emph{abandoned}.
A \emph{subgraph} induced by $S \subseteq V$ is denoted $G[S]$. A \emph{connected component} is a maximal connected subgraph.
A \emph{vertex cover} is a subset $V' \subseteq V$ such that every edge of $G$ is incident to at least one vertex in $V'$.

To analyze the complexity of Bus Factor measures, we recall the following classical \textit{NP-hard} decision problems~\cite{garey1979computers}:
\begin{description}
    \item \textsc{Set Cover:} Given a universe $U$, a family $\mathcal{F}$ of subsets of $U$, and an integer $k$, does there exist a subfamily $\mathcal{C} \subseteq \mathcal{F}$ with $|\mathcal{C}| \le k$ such that $\bigcup_{S \in \mathcal{C}} S = U$?
    \item \textsc{Clique:} Given a graph $G$ and integer $k$, does $G$ contain a fully connected subgraph of size at least $k$?
    \item \textsc{Vertex Cover:} Given a graph $G$ and integer $k$, does $G$ contain a vertex cover of size at most $k$? 
\end{description}

\paragraph{Unifying Existing Measures.}
We now formalize existing Bus Factor measures as combinatorial problems on the bipartite graph $B$.
Zazworka \textit{et al.}~\cite{zazworka2010developers} approached the problem by assessing \emph{redundancy}---asking: what is the largest set of people we can remove while keeping at least a fraction $\delta$ of tasks covered?
This leads to two distinct formulations based on whether we remove the least critical people (Best-Case) or the most critical people (Worst-Case):
\begin{align}
    Z_{best}(B, \delta) &= \max \{ k \mid \exists S \subseteq P, |S|=k : |\operatorname{Cov}(P \setminus S)| \ge \delta |T| \} \label{eq:zazworka_best} \\
    Z_{worst}(B, \delta) &= \max \{ k \mid \forall S \subseteq P, |S|=k : |\operatorname{Cov}(P \setminus S)| \ge \delta |T| \} \label{eq:zazworka_worst}
\end{align}
$Z_{best}$ essentially measures the maximum redundant workforce, while $Z_{worst}$ measures the project's resilience against the random loss of any $k$ employees.

In contrast, Cosentino~\emph{et al.}\cite{cosentino2015assessing} and Avelino~\emph{et al.}\cite{avelino2016novel} emphasized \emph{criticality}, defining the Bus Factor as the minimum number of people whose removal causes the project to stall (coverage drops below $\delta$).
We formalize this as:
\begin{equation}
    A(B, \delta) = \min \{ |S| : S \subseteq P \land |\operatorname{Cov}(P \setminus S)| < \delta |T| \}
    \label{eq:avelino}
\end{equation}

This formalization reveals a direct link between these seemingly distinct metrics. Specifically, it is straightforward to verify that $Z_{worst}(B, \delta) = A(B, \delta) - 1$.
That is, the Bus Factor ($A$) is simply the threshold at which the project's worst-case resilience ($Z_{worst}$) breaks.
This insight allows us to discard the specific ad-hoc heuristics used in prior work and subsume all definitions under two complementary combinatorial problems: \textsc{MRS} (Maximum Redundant Set) and \textsc{MCS} (Minimum Critical Set).
We define the problems below and provide proof sketches for their NP-hardness, referring the reader to~\cite{piccolo2026busfactor} for full formal proofs. Finally, we describe an effective approximation algorithm for each of them.

\begin{definition}[Maximum Redundant Set]
    Given a bipartite graph $B = (P, T, E)$ and $0 < \delta \le 1$, \textsc{MRS} seeks the largest set $S \subseteq P$ such that $|\operatorname{Cov}(P \setminus S)| \ge \delta |T|$.
\end{definition}

\begin{theorem}
    \textsc{MRS} is NP-hard.
\end{theorem}
\begin{proof}[Proof idea]
    This is equivalent to \textsc{Set Cover}. For $\delta = 1$, finding the largest set $S$ to remove is equivalent to finding the smallest set $X = P \setminus S$ to keep such that $X$ covers all tasks in $T$.
\end{proof}

\noindent\textbf{Approximation.} Since \textsc{MRS} is equivalent to \textsc{Partial Set Cover}, we employ the classic greedy strategy: iteratively select the professional $p \in P$ who covers the largest number of currently uncovered tasks, until the number of covered tasks is at least $\delta |T|$. The complement of this set is our approximate \textsc{MRS}.

\begin{definition}[Minimum Critical Set]
    Given a bipartite graph $B = (P, T, E)$ and $0 < \delta \le 1$, \textsc{MCS} seeks the smallest set $S \subseteq P$ such that $|\operatorname{Cov}(P \setminus S)| < \delta |T|$.
\end{definition}

\begin{theorem}
    \textsc{MCS} is NP-hard.
\end{theorem}

\begin{proof}[Proof idea]
    This can be proven via reduction from \textsc{Clique}. Given a graph $G=(V, E')$, we construct a bipartite graph $B=(P,T,E)$ where $P=V$, $T=E'$, and edges represent incidence.
    A task $t \in T$ (an edge in $G$) is \emph{abandoned} if and only if both its endpoints are removed.
    Thus, finding a set $S$ of size $k$ that causes the abandonment of $\binom{k}{2}$ tasks is equivalent to finding a clique of size $k$ in $G$.
    By setting $\delta$ such that the allowed remaining tasks are fewer than $|T| - \bigl(\binom{k}{2} - 1 \bigr)$, \textsc{MCS} solves \textsc{Clique}.
\end{proof}

\noindent\textbf{Approximation.} For \textsc{MCS}, we employ a greedy removal strategy based on node degree. In each step, we remove the professional $p \in P$ with the highest degree (number of active tasks).
We repeat this until the number of remaining covered tasks falls below $\delta |T|$. While simple, this heuristic aligns with the intuition that high-degree nodes (generalists) are critical failure points.

% \begin{definition}[Maximum Redundant Set]
%     Given a bipartite graph $B = (P, T, E)$ and $0 < \delta \le 1$, \textsc{MRS} seeks the largest set $S \subseteq P$ such that $\operatorname{Cov}(P \setminus S) \ge \delta |T|$.
% \end{definition}

% \begin{theorem}
%     \textsc{MRS} is NP-hard.
% \end{theorem}

% \begin{proof}[Proof idea]
%     Observe that for $\delta = 1$ \textsc{MRS} has to select the largest $S \subseteq P$ such that the complementary set $X = P \setminus S$ covers all the tasks in $T$. That is, $X$ is a minimum set cover for $B$.
% \end{proof}

% \begin{definition}[Minimum Critical Set]
%     Given a bipartite graph $B = (P, T, E)$ and $0 < \delta \le 1$, \textsc{MCS} seeks the smallest set $S \subseteq P$ such that $\operatorname{Cov}(P \setminus S) < \delta |T|$.
% \end{definition}

% \begin{theorem}
%     \textsc{MCS} is NP-hard.
% \end{theorem}

% \begin{proof}[Proof idea]
%     Observe that \textsc{MCS} can be used to answer the \textsc{Clique} problem. Given an undirected graph $G = (V,E')$ and an integer $k$, $G$ can be transformed into a bipartite graph $B = (P, T, E)$ by setting $P = V$, $T = E'$ and edges between nodes in $P$ and nodes in $T$ are given by the incidences in $G$. Note that a clique of size $k$ in $G$, if it exist, contains $k(k-1)/2$ edges. Now, on the constructed bipartite graph $B$ we compute \textsc{MCS} setting $\delta$ such that $\delta |T| = (k(k-1)/2) - 1$; that is, $\delta = (k(k-1)-2)/2|T|$. Now, it is easy to verify that $G$ has a clique of size $k$ if and only if this instance of \textsc{MCS} returns a set $S$ of size $k$.
% \end{proof}

\subsection{Proposed Measure: Bus Factor as Robustness}
We now introduce our novel measure, designed to address the structural limitations of prior formulations.
Existing measures like MCS and MRS treat professionals solely as resources, ignoring the \emph{topology} of collaboration.
Specifically, they fail to capture \emph{fragmentation}: the scenario where a project doesn't just stop, but shatters into isolated, non-communicating silos. Furthermore, they rely on arbitrary thresholds (e.g., $\delta=0.5$) that may not reflect the reality of a specific project.

To overcome these limitations, we redefine Bus Factor based on \emph{network robustness}~\cite{schneider2011mitigation}.
We monitor the size of the \textit{Largest Connected Task Component} (LCTC) as professionals are removed.
Let $\tau(B)$ denote the size of the LCTC in a graph $B$.
We define the robustness of a project $B=(P,T,E)$ under a specific removal sequence $\pi \in \Pi(P)$ of the professionals in $P$.
Let $B^{\pi}_{i}$ be the graph surviving after removing the first $i$ professionals in $\pi$.
We compute the area under the decay curve using the trapezoidal rule and normalize it against the theoretical maximum (a fully connected graph).

The normalized robustness for a removal sequence $\pi$ is:
\begin{equation}
    \mathcal{B}(B, \pi) =
    \frac{2}{|T|(2|P| - 1)}
    \sum_{i=1}^{|P|}
    \left[
        \tau\left(B^{\pi}_{i-1}\right) + \tau\left(B^{\pi}_{i}\right)
    \right]
\label{eq:normalised_bus_factor}
\end{equation}

This value $\mathcal{B} \in [0, 1]$ represents the percentage of the collaborative potential preserved during the removal sequence.
To capture the worst-case scenario, we define the project's Bus Factor as the robustness under the most destructive removal sequence: 
\begin{equation}
    \mathcal{B}_{BF}(B) = \min_{\pi \in \Pi(P)} \mathcal{B}(B, \pi)
\end{equation}

\paragraph{Computational Complexity.}
% Calculating the exact Robustness $\mathcal{B}_{BF}(B)$ requires finding the permutation $\pi$ that minimizes the area under the curve. We now show that this optimization problem is NP-hard by formulating its decision version.
In order to prove the hardness of the problem, we first formalize the structural shattering of a project using the concept of a \emph{Backbone Set}.

\begin{definition}[Backbone Set]
    Given a bipartite graph $B = (P, T, E)$, a set $S \subseteq P$ is a \emph{Backbone Set} if the subgraph induced by removing $S$ consists solely of disconnected stars, where each star is centered on a professional $p \in P \setminus S$.
\end{definition}

Intuitively, removing a backbone set destroys all collaborative paths between professionals, isolating every remaining person.
Note that for any removal sequence $\pi \in \Pi(P)$ the last removal leaves an empty graph, contributing zero to $\mathcal{B}(B, \pi)$.
Therefore, $\mathcal{B}(B, \pi)$ is fully determined by the first $|P|-1$ removals.
After these $|P|-1$ removals only one professional remains, and therefore the removed set forms a backbone set.
Consequently, computing $\mathcal{B}_{BF}(B)$ amounts to selecting a backbone set of $|P|-1$ professionals (and their order) that collapses the graph most efficiently.

We can thus formalize the robustness decision problem in terms of backbone sets:
\begin{definition}[Robustness Decision Problem]
    Given a bipartite graph $B = (P, T, E)$, an integer $k$, and a threshold $0 < M \le 1$, does there exist a backbone set of size $k$ with an associated removal sequence $\pi$ such that $\mathcal{B}(B, \pi) \le M$?
\end{definition}

\begin{theorem}
    The Robustness Decision Problem is NP-hard.
\end{theorem}

\begin{proof}[Proof Sketch]
    We reduce from \textsc{Vertex Cover}. Given an undirected graph $G = (V, E')$, we construct the bipartite incidence graph $B = (P, T, E)$ where $P = V$ and $T = E'$ (tasks represent edges).
    We prove hardness even for the special case where the threshold is vacuous (e.g., $M = 1$). In this case the decision problem reduces to determining whether a backbone set of size $k$ exists.

    $(\Rightarrow)$ If $G$ has a vertex cover $S$ of size $k \le |V|$, since $S$ covers all edges in $G$, removing $S$ in $B$ disconnects every task from at least one endpoint. The graph $B[P \setminus S]$ thus becomes a collection of disconnected stars (and therefore $S$ is a backbone set).

    $(\Leftarrow)$ Conversely, if no vertex cover of size $k$ exists in $G$, then for any set of $k$ professionals in $B$, there remains at least one edge in $G$ that is not covered. In $B$, this corresponds to a task $t$ that still connects two professionals. Thus, no backbone set of size $k$ exists in $B$.
    As such, solving the Robustness Decision Problem decides \textsc{Vertex Cover}.
\end{proof}

\paragraph{Approximation.}
Since the exact calculation is NP-hard, we approximate $\mathcal{B}_{BF}$ using a greedy strategy. We define the removal order $\pi$ by sorting professionals $p \in P$ in decreasing order of their degree.
To compute the area efficiently for large graphs, we employ a Union-Find data structure and simulate the removal process in reverse (starting from an empty graph and adding nodes back).
This allows us to track component sizes dynamically, ensuring the evaluation is nearly linear in the number of edges.
For implementation details, we refer the reader to our prior work~\cite{piccolo2024evaluating,piccolo2026busfactor}.

\section{Experimental Evaluation}
\label{sec:experiments}
While empirical evaluations on open-source repositories provide real-world context, they often lack the granularity required to rigorously assess the sensitivity of a metric.
In addition, many projects exhibit a trivial Bus Factor of 1 (as noted by Avelino et al.~\cite{avelino2016novel}), making them unsuitable for rigorous tests. 
To overcome this, and to evaluate $\mathcal{B}_{BF}$ under controlled conditions, we conduct a sensitivity analysis on synthetic bipartite graphs.

\paragraph{Experimental Setup.}
We generated a set of synthetic bipartite graphs exhibiting power-law degree distributions, consistent with the heavy-tailed nature of the workload distribution in real projects~\cite{yamashita2015pareto,agrawal2018hero,piccolo2018design,majumder2019software}.
Our baseline graphs consist of $|P|=7500$ professionals and $|T|=10000$ tasks.
We subject these graphs to controlled structural perturbations that mimic specific managerial actions (e.g., changing workload, hiring staff), and compare the behaviors of $\mathcal{B}_{BF}$, MCS, and MRS.
For MCS and MRS, we set the threshold $\delta = 0.5$, in continuity with prior research~\cite{avelino2016novel,zazworka2010developers}.
We investigate two research questions:
\begin{itemize}
    \item[] \textbf{RQ1:} How do metrics respond to changes in workload distribution?
    \item[] \textbf{RQ2:} How do metrics differentiate between adding specialists versus integrators?
\end{itemize}

\paragraph{Sensitivity to Workload Distribution (RQ1).}
We first investigate how the metrics respond to variations in task assignment density. We simulate this by adding and removing edges from the generated graphs.
Intuitively, a valid robustness measure should correlate positively with network density: more shared tasks imply higher redundancy.
We simulate this by adding and removing edges from the generated graphs in batches of 100, up to a total of 50000 modifications.
\begin{figure}[t]
    \centering
    \includegraphics[width=\linewidth]{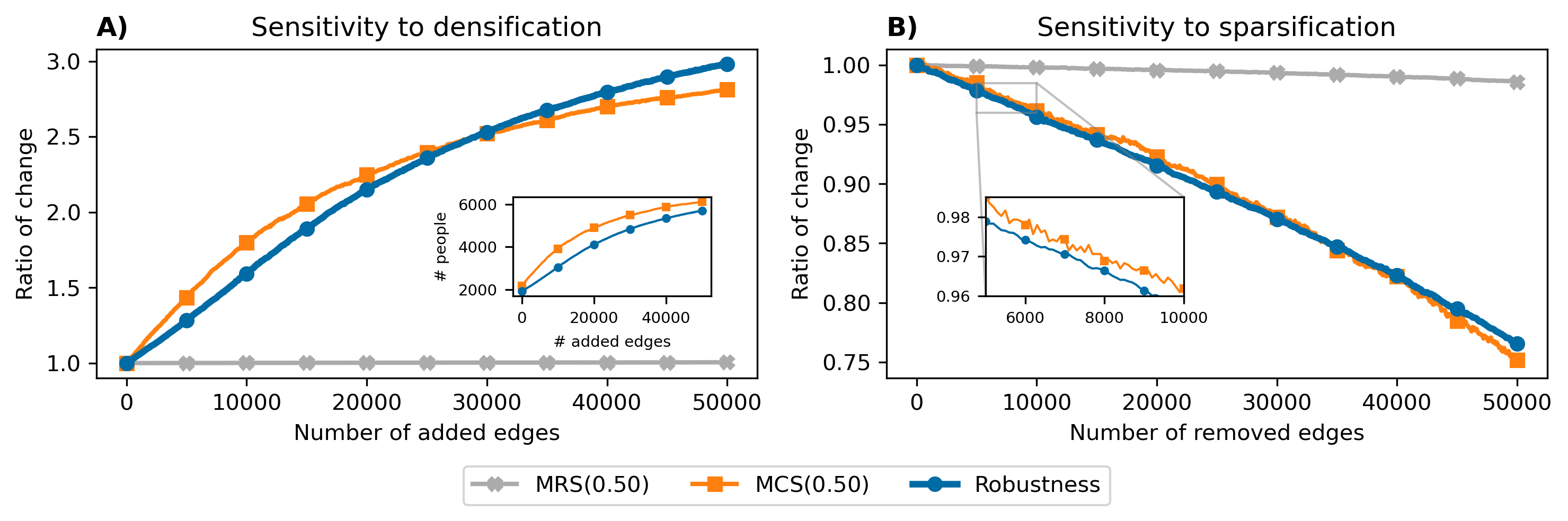}
    \caption{Sensitivity of Bus Factor measures to changes in network density (\textbf{RQ1}).
    \textbf{A)} Densification: MCS and $\mathcal{B}_{BF}$ increase with network density, as expected. Inset: values of MCS and $\mathcal{B}_{BF}$ expressed as number of people.
    \textbf{B)} Sparsification: MCS and $\mathcal{B}_{BF}$ decrease with network density, as expected. Inset: a zoomed-in view of the range $[5000, 10000]$ edges removed, showing that $\mathcal{B}_{BF}$ is more stable than MCS.}
    \label{fig:densification_sparsification}
\end{figure}
As shown in Figure~\ref{fig:densification_sparsification}, MRS is insensitive to changes in density. This is a consequence of its optimistic definition~\eqref{eq:zazworka_best}.
MCS and $\mathcal{B}_{BF}$, instead, exhibit the desired behavior, although MCS shows significant artifacts.
During densification (Fig.~\ref{fig:densification_sparsification}A), MCS saturates rapidly and plateaus, a direct consequence of its fixed coverage threshold.
Conversely, during sparsification (Fig.~\ref{fig:densification_sparsification}B), MCS exhibits oscillations as the removal of edges alters node degrees and shifts the critical set.
In contrast, $\mathcal{B}_{BF}$ scales smoothly and monotonically. 
Because $\mathcal{B}_{BF}$ measures connected components rather than simple coverage, it accurately reflects the gradual degradation or improvement of the project structural integrity without threshold-induced instability.

\paragraph{Sensitivity to Personnel Redundancy (RQ2).}
A robust metric must distinguish between headcount and structural resilience.
We simulate two distinct hiring strategies to test this distinction:
\begin{enumerate}
    \item \textit{Adding Singletons:} Adding specialists who work on a single task (one specialist per task).
    \item \textit{Adding Duplicates:} Cloning existing contributors (in decreasing order of degree) to increase backup potential.
\end{enumerate}

\begin{figure}[t]
    \centering
    \includegraphics[width=\linewidth]{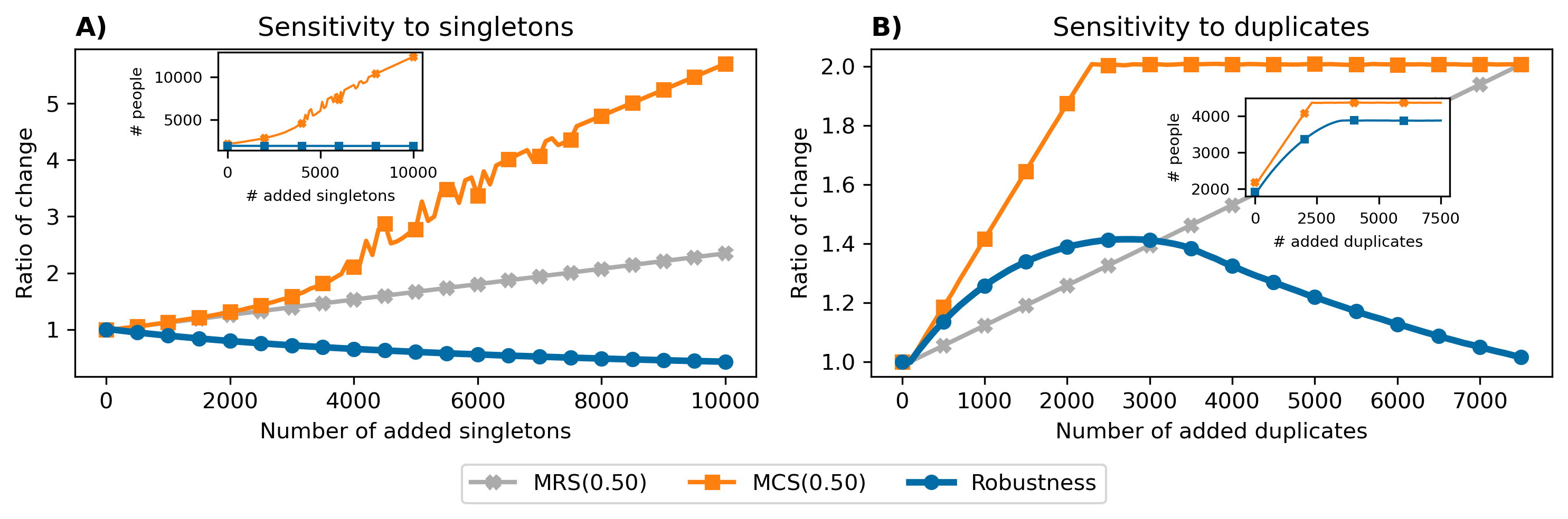}
    \caption{Sensitivity of Bus Factor measures to personnel redundancy (\textbf{Q2}).
    \textbf{A)} Sensitivity to singletons. \textbf{B)} Sensitivity to duplicates. In the inset: values of MCS and $\mathcal{B}_{BF}$ expressed as number of people.}
    \label{fig:singletons_duplicates}
\end{figure}

\noindent\textbf{The Singleton Problem.}
A critical flaw in coverage-based metrics is their inability to distinguish between resources and structure.
As shown in Figure~\ref{fig:singletons_duplicates}A, both MCS and MRS grow indefinitely as singletons are added. 
This implies that a project can be made infinitely robust simply by hiring isolated specialists, contradicting established project management theory~\cite{lawrence1986organization, heath2000coordination, piccolo2018design, piccolo2019iterations, majumder2019software}.
$\mathcal{B}_{BF}$ avoids this pitfall. It correctly identifies that singletons do not bridge structural gaps. 
In fact, $\mathcal{B}_{BF}$ decreases slightly, capturing the dilution of the project density and the diminishing returns of adding non-integrating personnel.

\noindent\textbf{Integrators vs. Specialists.}
When adding duplicates (Fig.~\ref{fig:singletons_duplicates}B), $\mathcal{B}_{BF}$ again demonstrates superior sensitivity. 
While MCS eventually saturates, $\mathcal{B}_{BF}$ increases rapidly when high-degree integrators are duplicated and plateaus as low degree people are added.
This confirms that our topological approach successfully distinguishes between structure and headcount, rewarding the creation of backups for central connectors while correctly identifying the diminishing returns of backing up peripheral people.

\section{Case Study: Structural Optimization in Practice}
\label{sec:casestudy}
To demonstrate the practical utility of $\mathcal{B}_{BF}$ beyond mere measurement, we apply it to a real-world project~\cite{piccolo2018design,piccolo2024evaluating}. 
Our goal is not only to evaluate the project's current structural robustness but to actively prescribe a more resilient task assignment matrix.

\paragraph{Baseline Evaluation and Permutation Test.}
We first evaluate the original personnel-task bipartite network, yielding a baseline robustness of 0.252. 
To understand if this value is an artifact of the specific degree distribution of the project, we generated a statistical ensemble of 10000 networks using a degree-preserving null model to perform a permutation test. (Figure~\ref{fig:optimized_bus_factor}B). 
Surprisingly, the original Bus Factor of the project is statistically \emph{lower} than what would be expected from a purely random assignment ($p = 0.007$). 
This indicates that the project suffers from severe silos that artificially inflate its vulnerability.

\paragraph{Optimization via Simulated Annealing.}
To correct this structural vulnerability, we treat the task assignment as an optimization problem. 
We employ a Simulated Annealing (SA) to iteratively rewire the bipartite graph in order to maximize $\mathcal{B}_{BF}$, while preserving the workload of every person. 
The computational feasibility is guaranteed by our linear-time Union-Find-based algorithm.

\begin{figure}
    \centering
    \includegraphics[width=\linewidth]{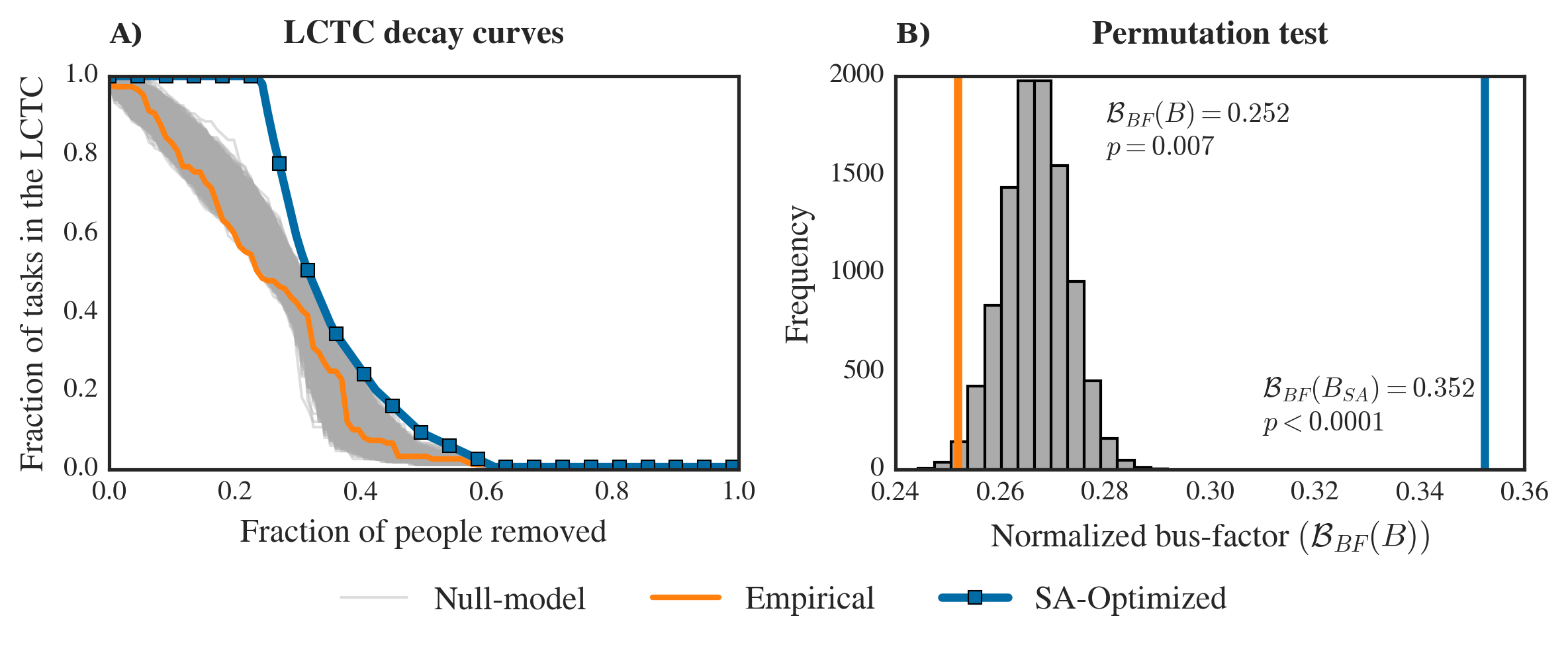}
    \caption{Comparison of the Bus Factor of a real-world project and its version optimized for robustness, against that of a statistical null model. \textbf{A}: decay curves. \textbf{B}: permutation test.}
    \label{fig:optimized_bus_factor}
\end{figure}

As shown in Figure~\ref{fig:optimized_bus_factor}A, the SA optimization shifts the decay curve drastically to the right. 
The optimized network, $B_{SA}$, achieves a robustness of $0.352$, a \textit{40\% improvement} over the original configuration.
The permutation test shows that this improvement is statistically significant (Figure~\ref{fig:optimized_bus_factor}B).

The magnitude of this optimization is clear in the decay curves. 
In the original network, removing the most connected 20\% of people reduces the size of the LCTC to 60\% of its original size.
In contrast, in the SA-optimized network, removing the exact same number of people barely disconnects the \emph{first} task. 
By merely reassigning existing workloads (without hiring new personnel), the project can be structurally hardened against turnover.

\section{Conclusion and Future Work}
\label{sec:conclusion}

In this paper, we challenged the prevailing resource-centric formulations of the Bus Factor. 
Traditional coverage-based measures, such as the Minimum Critical Set and Maximum Redundant Set, are structurally blind: they rely on arbitrary thresholds, fail to capture project fragmentation, and succumb to the \emph{Singleton Problem}.

To overcome these limitations, we introduced $\mathcal{B}_{BF}$, a novel metric grounded in bipartite network robustness. 
By monitoring the decay of the Largest Connected Task Component as people are removed, $\mathcal{B}_{BF}$ provides a normalized, continuous, topology-aware measure of collaborative resilience. 
While finding the optimal removal sequence is NP-hard, $\mathcal{B}_{BF}$ can be efficiently approximated in linear time.

Beyond accurate measurement, $\mathcal{B}_{BF}$ is highly actionable. 
As demonstrated in our case study, by utilizing $\mathcal{B}_{BF}$ as a fitness function within a Simulated Annealing framework, we successfully optimized the bus factor of a real-world project, improving its structural robustness by 40\% without adding new people or increasing their workload.

\paragraph{Future Work.}
This study paves the way for several promising research directions. 
First, we plan to extend our framework by allowing it to consider edge and task weights. 
Second, incorporating temporal network dynamics will enable us to track how project robustness evolves over its lifecycle, potentially serving as an early-warning system for structural decay before the departure of key personnel.

\begin{acknowledgments}
Sebastiano A. Piccolo is grateful to Marco Manna, Simona Perri, and Aldo Ricioppo for helpful comments. 
This research is partially supported by MUR under PNRR project PE0000013-FAIR, Spoke 9 - Green-aware AI -- WP9.2 and PN RIC project ASVIN ``Assistente Virtuale Intelligente di Negozio'' (CUP B29J24000200005).
\end{acknowledgments}

%% The declaration on generative AI comes in effect
%% https://ceur-ws.org/GenAI/Policy.html
\section*{Declaration on Generative AI}
During the preparation of this work, the authors used Gemini in order to: Grammar and spelling check. 
After using this tool, the authors reviewed and edited the content as needed and take full responsibility for the content of the publication.

%%
%% Define the bibliography file to be used
\bibliography{bibliography}

% \appendix


% The sources for the ceur-art style are available via
% \begin{itemize}
% \item \href{https://github.com/yamadharma/ceurart}{GitHub},
% % \item \href{https://www.overleaf.com/project/5e76702c4acae70001d3bc87}{Overleaf},
% \item
%   \href{https://www.overleaf.com/latex/templates/template-for-submissions-to-ceur-workshop-proceedings-ceur-ws-dot-org/pkfscdkgkhcq}{Overleaf
%     template}.
% \end{itemize}

\end{document}